\documentclass[a4paper,reqno]{amsart}

\usepackage{amsfonts,amssymb,amsthm,hyperref}

\newtheorem{Theorem}{Theorem}[section]

\newcommand{\stm}{\mathcal{M}}
\newcommand{\F}{\mathcal{F}}
\newcommand{\R}{\mathcal{R}}
\newcommand{\Id}{\mathcal{I}}
\newcommand{\Jd}{\mathcal{J}}
\newcommand{\tg}{\tilde{g}}
\newcommand{\hg}{\hat{g}}
\newcommand{\sgn}{\mathrm{sgn}}
\newcommand{\ie}{\textit{i.e.}}
\newcommand{\Rset}{\mathbb{R}}

\begin{document}

\title[Generic Metric Fields for RPS]{Equivalent Generic Forms for Metric Fields Yielded by Relativistic Positioning Systems}

\author{Jacques L. RUBIN}

\address{Universit{\'e} de Nice--Sophia Antipolis, UFR Sciences, Institut du Non-Lin{\'e}aire de Nice (UMR~7335), 
1361 route des Lucioles, F-06560 Valbonne, France}

\email{jacques.rubin@inln.cnrs.fr}

\date{\today}



\begin{abstract}
Relativistic positioning systems provide tensors represented in $\{\ell\ell\ell\ell\}$-frames ($\ell$ for light) dual to systems of emission coordinates. We show that any Lorentzian metric field given in such a frame is isometrically equivalent to a generic metric field defined by only four positive functions depending on another specific system of  emission coordinates.
\end{abstract}

\maketitle

\section{Introduction} 

Current positioning systems, such as GPS, GLONASS, Galileo, Beidou, or \linebreak IRNSS, are not true relativistic positioning systems, as is now well accepted. Indeed, the data they provide, such as, for instance, the time stamps broadcast by the satellites of these constellations and collected by the receiver devices of the user, cannot be sent or used directly without prior algorithmic corrections. The latter are due, in particular, to relativistic effects, undergone, for instance, by the on-board atomic clocks, and incorporated in the realtime computations (via the so-called Kalman filters) to obtain the ``correct'' spacetime positions. The relativistic effects taken into account in these processes \cite{ashby2003} are mainly the gravitational frequency shifts, the first and second-order relativistic Doppler shifts, and the Sagnac effect (even if the latter is not always considered a ``true'' relativistic effect by some authors).  Other corrections are included and based on models such as those for signal transmissions through the ionosphere or for the Earth's geoid. But, the latter are not really at the heart of the designs of these positioning systems, unlike the relativistic effects.
\par
To circumvent or avoid such fundamental root defects, new designs have been investigated and  
are thought of as providing true relativistic positioning systems (RPS). They are based on new protocols of spacetime positioning primarily devised, to our knowledge, by \text{B.~Coll}, \text{J.J.~Ferrando}, \text{J.A.~Morales} and \text{A.~Tarantola} \cite{Coll:2001zr,Coll:2002ly,collfermora06:dd,collfermorb06:ps}, initially in the case of two-dimensional spacetime and in four-dimensional spacetime with the SYPOR protocol for instance \cite{Coll:2003fj}. Moreover, \text{E.~Capolongo}, \text{M.L.~Ruggiero},  and \text{A.~Tartaglia} recently evaluated  such protocols for constructing the Earth's worldline with pulsars as celestial beacons \cite{Ruggiero:2011vn,tart10:ec,Tartaglia:2011kx,Tartaglia:2011ys}. Other approaches from pulsars has been proposed by Bunadar \textit{et al.} and Sheikh \textit{et al.} for evaluating the emission coordinates of an event \cite{Bunandar:2011fk,Sheikh:2006uq}.
\par
The coordinates generated by RPS are the so-called \textit{emission coordinates} given by the time stamps broadcast by the satellites, and then, their associated dual $\{\ell\ell\ell\ell\}$-frames are made of four future light-like basis vectors. The $\{\ell\ell\ell\ell\}$-frames yielded by RPS differ strongly from these dual $\{\ell\ell\ell\ell\}$-frames. Indeed, unlike the basis vectors of the dual $\{\ell\ell\ell\ell\}$-frames, the light-like basis vectors of the $\{\ell\ell\ell\ell\}$-frames yielded by RPS are the tangent vectors of the null geodesics travelled by the signals emitted by the satellites of the constellations. Along each of these null geodesics one time stamp (emission coordinate) is necessarily constant, and thus, in full generality, none of the dual light-like basis vectors can be tangent to such null geodesics. More precisely, if we denote by $g$ the Lorentzian metric field defined on the spacetime and $d\tau^i$ the differential 1-form associated to the emission coordinate $\tau^i$ ($i=1,\ldots,4$), then, $g$ can be written in the form: $g\equiv\sum_{i,j=1}^4g_{ij}(\tau)\,d\tau^i\,d\tau^j$, and any basis vector of a $\{\ell\ell\ell\ell\}$-frame yielded by a RPS is  the $g$-dual of a 1-form $d\tau^i$ rather than, merely, its dual $\tfrac{\partial\,\,}{\partial\tau^i}$. Also, $g$ can be represented in the dual $\{\ell\ell\ell\ell\}$-frame $(\tfrac{\partial\,\,}{\partial\tau^1},\ldots,\tfrac{\partial\,\,}{\partial\tau^4})$ by the matrix $G$ such that
\begin{equation} 
G \equiv
   \begin{pmatrix}
      0 & g_{12} & g_{13} & g_{14} \\
      g_{21} & 0 & g_{23} & g_{24} \\
      g_{31} & g_{32} & 0 & g_{34} \\
      g_{41} & g_{42} & g_{43} & 0
   \end{pmatrix},
   \label{metric}
\end{equation} 
where $g_{ij}=g_{ji}\neq0$ if $i\neq j$  ($i,\,j=1,\ldots,4$) and $\sgn(g_{ij})=-\varepsilon$ whenever the signature of $g$ is $2\varepsilon$ ($\varepsilon=\pm1$).
\par 
\text{B.~Coll} and \text{J.M.~Pozo} have made an extensive study of the algebraic properties of the class of metrics obtained specifically from RPS \cite{CollPozo06} . They have shown, in particular, that in the general case the terms $g_{ij}$ can never been factorized, apart, possibly, at very particular events in the spacetime, \ie, no set $\Lambda$ of 4 nonvanishing functions $\nu_i$ exists such that, for instance, $g_{ij}\equiv\nu_i\,\nu_j$ for all $i,\,j=1,\ldots,4$ such that $i\neq j$. 
\par
However, if $n\leqslant4$, we show that there always exists a set $\Lambda$ and a metric $\tg$ which is $\{\ell\ell\ell\ell\}$-equivalent to $g$, in a meaning to be specified in the sequel, with factorized components, \ie, there exists $n$ nonvanishing functions $\tilde{\nu}_i$ such that $\tg_{ij}\equiv\tilde{\nu}_i\,\tilde{\nu}_j$ for all $i,\,j=1,\ldots,n$ such that $i\neq j$. The number of nonvanishing components of $g $ and $\tg$ is the same, \ie, we have always $n(n-1)/2$ nonvanishing components out of the diagonal for either metric $g$ and $\tg$. The essential difference between $g$ and $\tg$ is that the $n(n-1)/2$ nonvanishing components of $\tg$ are not functionally independent.
\par
This equivalence is obtained from a change of local dual $\{\ell\ell\ell\ell\}$-frame related to a local change of emission coordinates. The metric $\tg$ can be exceptionally ascribed to a metric yielded by a RPS and, possibly, only at very particular event. In other words, no RPS yields a metric such as $\tg$ since the nonvanishing components should be functionally independent. Nevertheless, the $\{\ell\ell\ell\ell\}$-equivalence between $g$ and $\tg$ involves that the geometrical spacetime structure can be equivalently described by only $n$ functions rather than $n(n-1)/2$ functions.
\par
The two theorems we present below can be considered independently on the physical application in RPS although its interest in general relativity might be strongly relevant in complement of the theory of RPS.
\par
Also, the mathematical methods employed in the proofs of the two theorems are exhaustively indicated, in particular, in the reference \cite{BryantChGoldGrif:91} and they can be gathered under the designation of formal geometrical methods on the integrability of PDEs.

\section{The equivalent generic metric field} 

Let $\stm$ be a smooth connected $n$-dimensional pseudo-Riemannian manifold endowed with a Lorentzian metric $g$ represented as in \eqref{metric} in a given $\{\ell\ell\ell\ell\}$-frame defined on an emission coordinates chart $(U,\tau^1,\ldots,\tau^n)$ where the open $U\subset\stm$. We denote by $\partial_i$ the partial derivative with respect to the $i$-th emission coordinate $\tau^i$ of $\tau$. Then, the present paper is devoted to the proof of the following result:
\begin{Theorem}
If $n\leqslant4$, there always exists a smooth local diffeomorphism $f$ of which the Jacobian matrix is orthogonal and $n$ smooth positive functions $\nu_i$,  both defined on an open neighborhood $V\subset{U}$ of any given point of $U$, such that for all $\tau\equiv(\tau^1,\ldots,\tau^n)\in{V}$ the relations
\begin{equation} 
\tg_{ij}
   \equiv 
    \sum_{r,\,s=1}^n
    g_{rs}(f)(\partial_if^r)(\partial_jf^s)
    =\epsilon_{ij}\,\nu_i\,\nu_j\,,
\qquad
    i,j=1,\ldots,n\,,
\label{pdeg}
\end{equation} 
hold with $\epsilon_{ij}=\sgn(g_{ij})=\sgn(\tg_{ij})$ whenever $i\neq j$ and $\epsilon_{ij}=0$ otherwise. Then, we say that the ``generic'' metric $\tg$ is \emph{$\{\ell\ell\ell\ell\}$-equivalent} to $g$ (through $f$).
\end{Theorem}
Note that the non-diagonal terms of $g$ are not vanishing. 
Furthermore, if $n\leqslant3$, the result is trivial: take the identity map for $f$ and the functions $\nu_i$ are unique. Moreover, if $n=2$, we can make a separation of variables in $g_{12}$ such that each function $\nu_i$ ($i=1,2$) depends on only one emission coordinate (because any two-dimensional Riemann manifold is conformally flat) \cite{collfermora06:dd,collfermorb06:ps}. In cases of dimension greater than 4, some constraints on the definition of $g$ must be imposed. 
\par
The proof of this theorem presented below is made in the framework of the smooth category rather than the analytic category which is the standard situation for the application of the Cartan-K{\"a}hler theorem. 
Hence, no particular analytical criteria are discussed in relation to analytical boundary conditions for instance, and only the smooth Frobenius conditions are applied to check the integrability of the different PDEs involved in the proofs. Actually, we do not use either the Cartan-K{\"a}hler theorem or Cartan's test for involutivity. Hence, neither the computations of the codimensions of the polar spaces associated to the integral elements of certain flags nor the evaluations of their K{\"a}hler-regularities or regularities are performed \cite{BryantChGoldGrif:91}. The main reason is due to the non-standard way we ``transform'' a given set of algebraic equations defined in a jet bundle and associated to a system of PDEs  to an associated Pfaff system of contact 1-forms. We just make a little step aside in the definition of this ``transformation'' with strong advantages in the proof as a result, as there appear to be some often unnoticed forms of indetermination in the definition of the associated Pfaff system of a system of PDEs. 
\begin{proof} 
Let $\pi_n$ be the trivial fibration $\pi_n:\stm^2\equiv\stm\times\stm\longrightarrow\stm$, corresponding to the projection onto the first factor. We denote by $J_k(\pi_n)$ the fiber bundle of jets of order $k\geqslant0$ of the local smooth sections of $\pi_n$. In particular, we have $J_0(\pi_n)\equiv\stm^2$ with local coordinates $(\tau,\psi)\equiv(\tau^1,\,\ldots,\,\tau^n,\,\psi^1,\,\ldots,\,\psi^n)$. Furthermore, let 
\[
\psi_1\equiv(\tau^1,\,\ldots,\,\tau^n,\,\psi^1,\,\ldots,\,\psi^n,\,\psi^1_1,\,\psi_2^1,\,\ldots,\,\psi^i_j,\,\ldots,\,\psi^n_{n-1},\,\psi^n_n)
\]
be a local system of coordinates on $J_1(\pi_n)$. We denote also by $\Pi_k(\pi_n)\subset J_k(\pi_n)$ the set of invertible elements of $J_k(\pi_n)$, \ie, the set of $k$-jets of local smooth diffeomorphisms on $\stm$. $\Pi_k(\pi_n)$ is a groupoid with source map $\alpha_k:\Pi_k(\pi_n)\longrightarrow\stm$ where $\stm$ is the first factor of $\stm^2$ and the target map $\beta_k:\Pi_k(\pi_n)\longrightarrow\stm$ where we project onto the second factor.
Also, we denote by $\varPi_k(\pi_n)$ the presheaf of germs of local smooth $\alpha_k$-sections of $\Pi_k(\pi_n)$.
Then, we consider any solution of the system of PDEs \eqref{pdeg} as a sub-manifold of $\Pi_1(\pi_n)$ transversal to the $\alpha_k$-fibers and defined from the following system $\R_1$ of equations on the presheaf $\varPi_1(\pi_n)$:
\begin{equation}
\sum_{r,\,s=1}^ng_{rs}(\psi)\,\psi_i^r\,\psi_j^s-\epsilon_{ij}\,\nu_i\,\nu_j=0\,,
\qquad\qquad
i,j=1,\ldots,n\,,
\label{glambd}
\end{equation}
where $\nu_i>0$.
\par\medskip
Then, we denote also by $\hg_{ij}$ the terms such that
\[
\hg_{ij}\equiv\sum_{r,\,s=1}^ng_{rs}(\psi)\,\psi_i^r\,\psi_j^s\,,
\qquad\qquad
i,j=1,\ldots,n\,.
\]
Hence, $\R_1$ is also the following set of algebraic equations:
\[
\R_1:
\quad
\begin{cases}
\hg_{ii}=0\,,&\\
\hg_{ij}=\epsilon_{ij}\,\nu_i\,\nu_j\,,&\qquad i\neq j=1,\,\ldots,\,n\,.
\end{cases}
\]
We deduce easily for all distinct indices $i$, $j$ and $k$ that $\epsilon_{ij}\,\epsilon_{jk}\,\hg_{ij}\,\hg_{jk}=(\nu_j)^2\,\hg_{ik}\,\epsilon_{ik}$\,. Then, we must have
\begin{equation}
\R_1:\quad
\begin{cases}
\hg_{ii}=0\,,\\
|\hg_{ij}\,\hg_{jk}|=(\nu_j)^2\,|\hg_{ik}|\,,\quad\text{for all $i$, $j$ and $k$ distinct in $\{1,\ldots,n\}$.}
\end{cases}
\label{r1}
\end{equation}
In particular, if $n=4$ in \eqref{r1}, then, apart from the set of equations $\hg_{ii}=0$, the second set of equations are necessarily satisfied unless the two following deduced equations are not:
\[
|\hg_{12}\,\hg_{34}|=|\hg_{13}\,\hg_{24}|\,,
\qquad
|\hg_{13}\,\hg_{24}|=|\hg_{14}\,\hg_{23}|\,.
\]
Therefore, if $n=4$, the system $\R_1$ reduces to the following set of PDEs:
\[
\R_1:
\quad
\begin{cases}
\hg_{ii}=0\,,\qquad i=1,\ldots,4\,,&\\
|\hg_{12}\,\hg_{34}|=|\hg_{13}\,\hg_{24}|\neq0\,,&\\
|\hg_{13}\,\hg_{24}|=|\hg_{14}\,\hg_{23}|\neq0\,.&
\end{cases}
\]
Rewriting this system of PDEs without the absolute values, we obtain
\begin{equation}
\R_1:
\quad
\begin{cases}
\F_i(\psi_1)\equiv\sum_{r,s=1}^4g_{rs}(\psi)\,\psi^r_i\,\psi^s_i=0\,,\qquad i=1,\ldots,4\,,&\medskip\\
\F_5(\psi_1)\equiv\sum_{i,j,k,h=1}^4g^\epsilon_{ijkh}(\psi)\,\psi^i_1\,\psi^j_2\,\psi^k_3\,\psi^h_4=0\,,&\medskip\\
\F_6(\psi_1)\equiv\sum_{i,j,k,h=1}^4g^{\epsilon'}_{ijhk}(\psi)\,\psi^i_1\,\psi^j_2\,\psi^k_3\,\psi^h_4=0\,,&
\end{cases}
\label{r1phi}
\end{equation}
where
\[
g^\epsilon_{ijkh}\equiv g_{ij}\,g_{kh}-\epsilon\,g_{ik}\,g_{jh}\,,
\qquad\qquad
\epsilon=\pm1\,,
\quad
i,j,k,h=1,\ldots,4\,.
\]
Before going further, we must know if there exist solutions to the system of homogeneous polynomial equations \eqref{r1phi} in the variables $\psi^i_j$ whenever $\psi$ is fixed. First, we denote by $\phi_i$  the linearly independent column vectors such that $\phi_i\equiv(\psi^j_i)$. From the first four equations $\F_i(\psi_1)=0$ ($i=1,\,\ldots,:,4$), the former must be light-like vectors which still exist since $g$ is Lorentzian. Second, the last two functions can be rewritten as $\F_5(\psi_1)=g(\phi_1,\phi_5)$ and $\F_6(\psi_1)=g(\phi_1,\phi_6)$ where
$\phi_5\equiv\hg_{34}\,\phi_2-\epsilon\,\hg_{24}\,\phi_3$ and $\phi_6\equiv\hg_{34}\,\phi_2-\epsilon'\hg_{23}\,\phi_4$.
Therefore, the nonvanishing vectors $\phi_5$ and $\phi_6$ are collinear to $\phi_1$ or time-like ($\hg_{ij}\neq0$ if $i\neq j$). However, because the four vectors $\phi_i$ ($i=1,\,\ldots,\,4$) are linearly independent, then $\phi_5$ and $\phi_6$ must be time-like. Hence, the signs of their norms $g(\phi_5,\,\phi_5)$ and $g(\phi_6,\,\phi_6)$ are equal to the sign of the signature $2\varepsilon$ of $g$, \ie, we have
\[
\sgn(\hg_{34}\,\hg_{24}\,\hg_{23})=-\epsilon\varepsilon=-\epsilon'\varepsilon\,.
\]
Thus, in particular, we must have $\epsilon'=\epsilon$ in the system \eqref{r1phi}. Besides, $\epsilon$ is arbitrary, and then,  from now and throughout, we set also $\epsilon=\varepsilon$. As a result, we have solutions to the system \eqref{r1phi} if and only if\,\footnote{Note that this inequality illustrates the first form of the Tarski-Seidenberg theorem \cite{BochCostRoy:98,CosteRAAG}.}
\begin{equation}
\hg_{34}\,\hg_{24}\,\hg_{23}<0\,.
\label{3hg}
\end{equation}
Next, we consider the expression $\hg_{34}\,\hg_{24}\,\hg_{23}$ as a quadratic form $Q$ with respect to $\phi_2$. We obtain $\hg_{34}\,\hg_{24}\,\hg_{23}=Q(\phi_2,\,\phi_2)\equiv\sum_{i,j=1}^4 Q_{ij}\,\phi^i _2\,\phi^j_2$ where 
\[
Q_{ij}
=\left(\sum_{h,\,k=1}^ng_{rs}(\psi)\,\phi_3^h\,\phi_4^k\right)\left(\sum_{s=1}^4g_{js}(\psi)\,\phi^s_3\right)\left(\sum_{r=1}^4g_{ir}(\psi)\,\phi^r_4\right),
\] 
and then, the inequality \eqref{3hg} is always satisfied if $Q$ is not a positive elliptic form. For, it suffices that one of the diagonal terms $Q_{ii}$ to be non-positive since, in this case, it implies the existence of basis vectors of non-positive norms with respect to $Q$ if $Q$ is non-degenerate.\footnote{We can use also the Coll-Morales rules \cite[see \S\,III]{collmor93} generalizing more effectively the Jacobi, Gundelfinger and Frobenius rules with the notion of \textit{causal sequence}  $(i_1,\,i_2,\,i_3)\equiv(\sgn(\triangle_1),\,\sgn(\triangle_2),\,\delta\,\sgn(\triangle_3))$ where the $\triangle_k$'s are the first three \textit{leading principal minors} of $Q$ of order $k$ and $\delta$ is the \textit{determinant index}. In the present case, the causal sequence should differ from the causal sequence $(1,\,1,\,1)$.} We cannot ensure in full generality the non-degeneracy of $Q$, and thus, we impose, in particular, the condition $Q_{11}<0$ only and not the condition $Q_{11}=0$. Proceeding in the same way, the term $Q_{11}$ is still considered as a quadratic form $R$ with respect to $\phi_3$. And again, we have $R(\phi_3,\phi_3)\equiv Q_{11}<0$ if $R$ is not a positive elliptic form. For the same reasons as above for $Q_{11}$,
this condition is always satisfied, in particular, if there exists a diagonal term $R_{ii}$ such that $R_{ii}<0$.
We consider the term $R_{22}$ such that
\begin{equation}
R_{22}\equiv{}g_{12}(\psi)\left(\sum_{i=1}^4g_{1i}(\psi)\,\phi^i_4\right)\left(\sum_{j=1}^4g_{2j}(\psi)\,\phi^j_4\right).
\label{r22}
\end{equation}
Then, because $g$ is non-degenerate, the coefficients $g_{2j}$ and $g_{1j}$ for $j=1,\ldots,4$ cannot be proportional. Therefore,  the two hyperplanes in $\Rset^4$ defined by the two last factors in \eqref{r22} and linear with respect to $\phi_4$ are strictly distinct. As a result, $\Rset^4$ is divided by these two hyperplanes into four connected open subsets. Then, we can always find a vector $\phi_4\in\Rset^4$ in one of these four subsets such that $\left(\sum_{i=1}^4g_{1i}(\psi)\,\phi^i_4\right)\left(\sum_{j=1}^4g_{2j}(\psi)\,\phi^j_4\right)$ has the opposite sign of $g_{12}(\psi)(\neq0)$ and thus such that $R_{22}<0$.
Hence, there always exist real solutions to the system \eqref{r1phi} whatever are the source $\tau$ and the target $\psi$. 
\par\smallskip
Additionally, from the inequality \eqref{3hg} and the `continuity of roots' property\footnote{The `continuity of roots' property ensures the roots of a given finite set of algebraic equations to be continuously depending on the coefficients parameterizing these algebraic equations.}  \cite[p.~363]{whitney72}, we deduce that, given a point $\psi$, there always exists a maximal open subset $U_\psi\subset\stm$ of $\psi$ such that this set of solutions $S_\psi$ is always an open smooth manifold of \textit{constant} dimension at least 10 on $U_\psi$\,. 
As a result, $U_\psi$ is also necessarily closed, but then, because $\stm$ is connected, we deduce that $\dim S_\psi=m$ is a constant on $\stm$.
Moreover, $\alpha_1\times\beta_1$ is a surmersion on $\stm^2$, and thus, the latter has no critical points in $\R_1$. Therefore, we obtain that $m=10$ \cite[see Lemma 1, p.11]{MilnorTopo:97}.
\par\smallskip
It follows that the restrictions to $\R_1$ of the source and target maps are surmersions, and then, the system $\R_1$ is, respectively, \textit{formally integrable} (as a system of local diffeomorphisms defined on the whole of $\stm$), and  \textit{homogeneous} (transitive diffeomorphisms from opens to any other opens in $\stm$). And then, $\R_1$ is a differentiable manifold such that $\dim\R_1=18$.
\par
Next, we consider the following canonical contact structure $S_0$ of width $n$ (\ie, $n$-flag  \cite{KumpRub:02,MormulMulti:04}) and length $1$ on $\Pi_1(\pi_n)$ generated by the set $\left\{\omega^1,\,\omega^2,\ldots,\,\omega^n\right\}$ of contact 1-forms  $\omega^i\in T^*\!J_1(\pi_n)$ such that
\begin{equation}
S_0:\quad
\begin{cases}
\omega^1=&d\psi^1-\sum_{i=1}^n\psi^1_i\,d\tau^i\,,\medskip\\
\omega^2=&d\psi^2-\sum_{j=1}^n\psi^2_j\,d\tau^j\,,\medskip\\
\ldots=&\dotfill\,,\,\,\medskip\\
\omega^n=&d\psi^n-\sum_{k=1}^n\psi^n_k\,d\tau^k\,.
\end{cases}
\label{s0}
\end{equation}
Obviously, the \textit{terminal system} $S_1$ of $S_0$ is vanishing \cite{KumpRub:02}.
Then, we complement the set of contact 1-forms generating $S_0$ with another set of 1-forms $\omega^i_j$ on $\Pi_1(\pi_n)$ defined by the relations:
\begin{equation}
\omega^i_j\equiv d\psi^i_j-\sum_{k=1}^nz^i_{jk}(\psi_1)\,d\tau^k\,,
\qquad\qquad
i,j=1,\ldots,n\,,
\label{omij}
\end{equation}
where any given set of functions $z^i_{jk}(\psi_1)\in C^\infty(\Pi_1(\pi_n))$ (with $i,j,k,h=1,\ldots,n$) must satisfy
\begin{equation}
z^i_{jk}(\psi_1)=z^i_{kj}(\psi_1)\,,
\qquad\qquad
D_kz^i_{jh}(\psi_1)=D_hz^i_{jk}(\psi_1)\,,
\label{dz}
\end{equation}
where $D_k$ is the formal differentiation with respect to $\tau^k$ defined by the formula
\[
D_k\equiv
\frac{\partial\,\,\,\,}{\partial\,\tau^k}
+\sum_{i=1}^n\psi^i_k\,\frac{\partial\,\,\,\,}{\partial\,\psi^i}
+\sum_{i,j=1}^nz^i_{jk}(\psi_1)\,\frac{\partial\,\,\,\,}{\partial\,\psi^i_j}\,,
\qquad\qquad
k=1,\ldots,n\,.
\]
From this definition and for any smooth function $\F$ defined on $J_1(\pi_n)$ we find that the commutator $[D_k,\,D_h]$ satisfies the relation
\begin{equation}
[D_k,\,D_h](\F)=\sum_{i,j=1}^n\left(D_kz^i_{jh}-D_hz^i_{jk}\right)\frac{\partial\,\F}{\partial\,\psi^i_j}\,.
\label{comD}
\end{equation}
Then, we denote by $T_0(z)\supseteq S_0$ this new contact structure generated by the contact 1-forms $\omega^i$ and the 1-forms $\omega^i_j$ ($i,j=1,\ldots,n$). 
In particular, from \eqref{comD} and the relation $d^2\omega=0$ for any smooth $p$-forms $\omega$ in $\Lambda T^*J_1(\pi_n)$, we deduce also that the \textit{Martinet structure tensor} $\delta\equiv d \mod T_0(z)$ is such that $\delta^2=0$ \cite{MartinetDelta:74}.
\par
Moreover, from relations \eqref{s0} and \eqref{omij}, we obtain:
\[
\begin{cases}
\,d\omega^i=\sum_{k=1}^nd\tau^k\wedge\omega^i_k\,,&\medskip\\
\,d\omega^j_k=
\sum_{h,r,s=1}^n\left(\dfrac{\partial z^j_{kh}}{\partial\psi^r_s}\right)d\tau^h\wedge\omega^r_s
+\sum_{h,r=1}^n\left(\dfrac{\partial z^j_{kh}}{\partial\psi^r}\right)d\tau^h\wedge\omega^r\,,&
\end{cases}
\]
and then, $T_0(z)$ satisfies the Frobenius conditions (equivalent to $\delta^2=0$) and is an integrable Pfaff system on $\Pi_1(\pi_n)$.
\par
Next, we consider $\R_1$ as a presheaf $\Id_1$ of ideals locally finitely generated  by the functions $\F_i$ ($i=1,\dots,6$) defined on $\Pi_k(\pi_4)$ and we assume that any manifold on which this presheaf vanishes, \ie, the sub-manifold defined from a solution, is an integral sub-manifold of a $T_0(z)$ in $\Pi_k(\pi_4)$. We denote by $V_1(z)$ the foliation of all of these integral sub-manifolds. This latter version conforms better with the classical concepts of integral manifolds and differs from the approach of PDEs translated in terms of presheafs of Pfaff systems of contact 1-forms satisfying the Frobenius conditions (see for instance \cite{Bottint:70}). 
\par
As a consequence, denoting by $\Jd_1(z)$ the presheaf of differential ideals generated by $T_0(z)$ on $J_1(\pi_4)$, we say that $\R_1$ is integrable on $\stm$ if there exists a sub-manifold of solutions $V_1(z)\subseteq \Pi_1(\pi_4)$ and a nonvanishing presheaf $\Jd_1(z)$ such that $\Id_1\subseteq\Jd_1(z)$ on $V_1(z)$.
\par
In other words, if a set of functions $z^i_{jk}$ exists satisfying the latter condition, a smooth local diffeomorphism $f$ of $\stm$ is a solution of $\R_1$ if and only if
\[
\begin{cases}
\F_i(j_1(f))=\iota_0\,,&\qquad i=1,\dots,6\,,\\
f^*(\omega^i)=0\,,&\qquad
f^*(\omega^j_k)=0\,,\qquad i,j,k=1,\ldots,4\,,
\end{cases}
\]
where $\iota_0$ is the zero function on $\stm$ and $j_1(f)$ is the first prolongation of $f$; and thus a local section of $\Pi_1(\pi_4)$. Hence, from \eqref{omij}, we obtain that
\[
\begin{cases}
df^i=\sum_{k=1}^4(\partial_kf^i)\,d\tau^k\,,&i=1,\ldots,4\,,\\
df^j_k=\sum_{h=1}^4z^j_{kh}(j_1(f))\,d\tau^h\,,&j,k=1,\ldots,4\,.
\end{cases}
\]
And then, from the second order of derivation and from the successive prolongations, all of the derivatives of $f$ are functionals of the derivatives of $f$ of order less than or equal to one. As a result, a Taylor expansion for $f$ can be deduced with Taylor coefficients defined from the Taylor coefficients of $f$ of order less than or equal to one only. Thus, we obtain a formal Taylor expansion for $f$ which can be convergent on a suitable relatively compact open neighborhood $U_\tau$ of any point $\tau\in\stm$ if some Lipschitzian conditions on the functions $z^j_{kh}$ are satisfied on $(\alpha_1)^{-1}(U_\tau)\cap V_1(z)$; justifying the definition of integrability given above for $\R_1$.
\par
To satisfy the condition $\Id_1\subseteq\Jd_1(z)$ on a manifold $V_1(z)$, we must have $d\F_i\equiv0\mod T_0(z)$ for all of the indices $i=1,\ldots,6$ on $V_1(z)$. We obtain the following system $\mathcal{S}(z)$ of 24 linear equations with 24 unknowns $z^i_{jk}$:
\begin{multline*}
\delta\F_i=0\Longrightarrow\sum_{r,s=1}^4\left(
\left(\sum_{k=1}^4(\partial_kg_{rs})(\psi)\,\psi^k_h\right)\psi^r_i\,\psi^s_i
+g_{rs}(\psi)\,\psi^r_i\,z^s_{ih}
\right)=0\,,\hfill
\end{multline*}
\begin{multline*}
\delta\F_5=0\Longrightarrow
\sum_{i,j,k,h,r=1}^4(\partial_rg^\varepsilon_{ijkh})(\psi)\,\psi^i_1\,\psi^j_2\,\psi^k_3\,\psi^h_4\,\psi^r_s\\
+\sum_{i,j,k,h=1}^4g^\varepsilon_{ijkh}
\left\{
\psi^i_1\,\psi^j_2\,\psi^k_3\,z^h_{4s}+\psi^i_1\,\psi^j_2\,z^k_{3s}\,\psi^h_4
\right.\\
\left.
+\,\psi^i_1\,z^j_{2s}\,\psi^k_3\,\psi^h_4+z^i_{1s}\,\psi^j_2\,\psi^k_3\,\psi^h_4\right\}=0\,,
\end{multline*}
\begin{multline*}
\delta\F_6=0\Longrightarrow
\sum_{i,j,k,h,r=1}^4(\partial_rg^{\varepsilon}_{ijhk})(\psi)\,\psi^i_1\,\psi^j_2\,\psi^k_3\,\psi^h_4\,\psi^r_s\\
+\sum_{i,j,k,h=1}^4g^{\varepsilon}_{ijhk}
\left\{
\psi^i_1\,\psi^j_2\,\psi^k_3\,z^h_{4s}+\psi^i_1\,\psi^j_2\,z^k_{3s}\,\psi^h_4
\right.\\
\left.
+\,\psi^i_1\,z^j_{2s}\,\psi^k_3\,\psi^h_4+z^i_{1s}\,\psi^j_2\,\psi^k_3\,\psi^h_4
\right\}=0\,.
\end{multline*}
Note that if $n>4$ we have more equations than unknowns, and then not all metric fields  $g$ are admissible to satisfy the conditions of the theorem.
Now, setting for all of the functions $z^i_{jk}$ the relations
\begin{equation}
z^i_{jk}(\psi_1)=\psi^i_j\,\sum_{h=1}^4\psi^h_k\,z_h(\psi)\,,
\label{zijk}
\end{equation}
where the functions $z_k$ depend on $\psi$, we find that the unique solution of $\mathcal{S}(z)$ is the set of functions 
$z^i_{jk}$ such that
\begin{equation}
z_k(\psi)=-\frac{1}{8}\sum_{i,j=1}^4g^{ij}(\psi)\,(\partial_kg_{ij})(\psi)\equiv-\frac{1}{4}\sum_{i=1}^4\Gamma^i_{ik}(\psi)\,,
\label{zk}
\end{equation}
where the $\Gamma^i_{jk}$ are the Christoffel symbols of $g$. Then, it remains to see that the conditions \eqref{dz} are satisfied.
For, we must have the relations
\[
\sum_{h=1}^4\left(z^i_{jr}(\psi_1)\,\psi^h_k-z^i_{jk}(\psi_1)\,\psi^h_r\right)z_h(\psi)=0\,,
\]
which are, actually, verified with the functions $z^i_{jk}$ given by the relations \eqref{zijk} with \eqref{zk}.
Moreover, because no algebraic constraints exist on $\psi_1$, apart from those obtained from the vanishing of the functions $\F_i$ which are elements of $\Id_1\subseteq\Jd_1(z)$, then the manifold $V_1(z)$ is the whole of the open manifold $\Pi_1(\pi_4)$. Furthermore, the 1-forms $\omega^k$ and $\omega^i_j$ are the so-called \textit{basic 1-forms} \cite{MolinoBasic:77} associated with any \textit{complete transversally parallelizable} foliation. Lastly, at any given point $\tau$, the finite system of equations \eqref{r1phi} in the variables $\psi_1$ always have solutions, and the set of positive functions $\nu_i$ is not unique. 
\end{proof}
Besides, we note that $\R_1$ is not a Lie groupoid \cite{Mackenzie:1987fk} because if $g$ is $\{\ell\ell\ell\ell\}$-equivalent to $\epsilon_{ij}\,\nu_i\,\nu_j$ and $\epsilon_{ij}\,\tilde\nu_i\,\tilde\nu_j$ through, respectively, the diffeomorphisms $f$ and $\tilde{f}$, then, there may not always be four positive functions  $\hat\nu_i$ such that  $g$ would be $\{\ell\ell\ell\ell\}$-equivalent to $\epsilon_{ij}\,\hat\nu_i\,\hat\nu_j$ through  $f\circ\tilde{f}$ or $\tilde{f}\circ f$. Nevertheless, we have an associated \textit{principal groupoid} regarded as the graph of the $\{\ell\ell\ell\ell\}$-equivalence relation, and then the equivalence class $[g]$ of the given metric $g$ is a source fiber in this groupoid. Moreover, if $\stm$ is time oriented, \ie, there exists a complete (future time-like) vector field $\xi$ on $\stm$, then, because $\R_1$ is also a differentiable $\alpha_1$-fiber bundle, we also have in the smooth category the following:
\begin{Theorem}
If $n=4$, then, given a Lorentzian metric $g$ on $\stm$ assumed to be time oriented, connected and simply connected, then, there exists only one smooth diffeomorphism $f^i(\tau)\equiv\psi^i$ being a solution of $\R_1$ of which the Jacobian matrix is an element of $SO(4)$; and, as a result, there is a unique set of four positive functions $\nu_i$.
\end{Theorem}
\begin{proof}
Let $\psi_0\equiv(\tau,\,\psi)$ be any  point in $\stm^2$ and a matrix $\Psi\equiv(\psi^i_j)\in\alpha_1^{-1}(\tau)\times\beta_1^{-1}(\psi)\equiv\R_1^{\psi_0}$. Then, in particular, we have $\det\Psi\neq0$\,, and from the precedent proof  we have also $\dim\R_1^{\psi_0}=10$.  Let  $\psi_0$ be a fixed point, then the coefficients $\psi^i_j$ of $\Psi$ satisfy a system consisting of the six homogeneous equations \eqref{r1phi}. If, moreover, the four column vectors $\phi_k\equiv(\psi_k^i)$ ($k=1,\,\ldots,4$) are orthogonal each to the others, then, additionally, $\Psi$ verifies a system consisting of six multivariate quadratic equations $Q_i(\Psi)=0$  ($i=1,\,\ldots,\,6$) (the six scalar products of the four column vectors $\phi_k$). Hence, let $r_{\psi_0}$ be the smooth map such that $r_{\psi_0}:\Psi\in\R_1^{\psi_0}\longrightarrow(Q_1(\Psi),\,\ldots,\,Q_6(\Psi))\in\mathbb{R}^{6}$\,, then, we can show that $\ker r_{\psi_0}$ is a nonempty four dimensional manifold \cite[see Lemma 1, p.11]{MilnorTopo:97}. Indeed, the tangent map $Tr_{\psi_0}$ is regular in $\R_1^{\psi_0}$ because the coefficients of $Tr_{\psi_0}$ are linear with respect to $\Psi$, and then, if $\det Tr_{\psi_0}=0$, we would have the four vectors $\phi_k$ not linearly independent, which is not possible from the relation $\det\Psi\neq0$. In addition, because the twelve polynomials $Q_i$ and $\F_j$ are homogeneous, then the four vectors $\phi_k$ can be normalized, and thus, $\Psi\in{}SO(4)$. It follows that 1) $S^{\psi_0}\equiv{}SO(4)\cap\R_1^{\psi_0}$ is not empty, and 2) $S^{\psi_0}$ is a real semialgebraic set consisting of sixteen homogeneous multivariate polynomial equations of even degrees and one inequation. Consequently, because there are as many algebraic equations than unknowns, we obtain a nonempty \textit{finite} \cite[\S~2.3]{BochCostRoy:98,CosteRAAG} set $s(\psi_0)$ of real roots $\Psi\in SO(4)$ which are solutions of the system \eqref{r1phi}. Moreover, from the `continuity of roots' property, the continuity of $g$ on $\stm$ and the connexity of $\stm^2$, we deduce that $|s(\psi_0)|$ is constant over $\stm^2$; And then, the set $\cup_{\psi_0\in\stm^2} S^{\psi_0}$ is a covering of $\stm^2$ which is universal because $\stm^2$ is simply connected. 
Therefore, there is only one preimage of $\psi_0$ under $\alpha\times\beta$ in $S^{\psi_0}$. 
\end{proof}

\end{document}